\documentclass[11pt, a4paper]{article}
\usepackage[left=25mm,right=25mm,top=25mm,bottom=30mm]{geometry}
\usepackage{framed,latexsym,amssymb,amsmath,amsthm,float,cite,complexity,setspace,color,graphicx,lineno,hyperref,url,lmodern}
\usepackage[shortlabels]{enumitem}
\usepackage[ddmmyyyy]{datetime}
\usepackage[noline,noend]{algorithm2e}
\usepackage{tikz,pgfplots}
\pgfplotsset{compat=1.18}

\usepackage{thmtools,thm-restate}
\usepackage{regexpatch}
\makeatletter
\xpatchcmd\thmt@restatable{%
\csname #2\@xa\endcsname\ifx\@nx#1\@nx\else[{#1}]\fi
}{%
\ifthmt@thisistheone
\csname #2\@xa\endcsname\ifx\@nx#1\@nx\else[{#1}]\fi
\else
\csname #2\@xa\endcsname[{restated}]
\fi}{}{}
\makeatother

\newtheorem{theorem}{Theorem}
\newtheorem{lemma}[theorem]{Lemma}
\newtheorem{corollary}[theorem]{Corollary}
\newtheorem{claim}{Claim}
\newtheorem{proposition}[theorem]{Proposition}
\theoremstyle{definition}

\floatstyle{boxed}
\newfloat{problem}{tbhH}{prb}

\tikzset{
	dot/.style = {circle, fill, minimum size=#1,
		inner sep=0pt, outer sep=0pt},
	dot/.default = 5pt
}
\usetikzlibrary{calc}
\usetikzlibrary{tikzmark}

\newcommand{\calH}{\mathcal{H}}
\newcommand{\calL}{\mathcal{L}}
\newcommand{\calM}{\mathcal{M}}
\newcommand{\calR}{\mathcal{R}}
\newcommand{\calF}{\mathcal{F}}
\newcommand{\calO}{\mathcal{O}}

\begin{document}
\onehalfspace

\title{A Faster Algorithm for Independent Cut}
\author{Vsevolod Chernyshev$^1$ \and Johannes Rauch$^1$ \and Dieter Rautenbach$^1$ \and Liliia Redina$^2$}
\date{}
\maketitle
\vspace{-10mm}
\begin{center}
{\small 
$^1$ Institute of Optimization and Operations Research, Ulm University, Ulm, Germany\\
$^2$ HSE University, Moscow, Russia\\
\texttt{\{vsevolod.chernyshev, johannes.rauch, dieter.rautenbach\}@uni-ulm.de}, \texttt{liliiaredina@gmail.com}
}
\end{center}

\begin{abstract}
\noindent
The previously fastest algorithm for deciding the existence of an independent cut had a runtime of $\calO^*(1.4423^n)$, where $n$ is the order of the input graph.
We improve this to $\calO^*(1.4143^n)$.
In fact, we prove a runtime of $\calO^*\left( 2^{(\frac{1}{2}-\alpha_\Delta)n} \right)$ on graphs of order $n$ and maximum degree at most $\Delta$, where $\alpha_\Delta=\frac{1}{2+4\lfloor \frac{\Delta}{2} \rfloor}$.
Furthermore, we show that the problem is fixed-parameter tractable on graphs of order $n$ and minimum degree at least $\beta n$ for some $\beta > \frac{1}{2}$, where $\beta$ is the parameter.
\end{abstract}

\section{Introduction}
We consider finite, simple, and undirected graphs
and use standard terminology and notation.
Let $G$ be a graph and let $S$ be a set of vertices of $G$.
The set $S$ is \emph{independent} in $G$ 
if no edge of $G$ connects two vertices from $S$.
The set $S$ is a \emph{cut} of $G$ if $G-S$ is disconnected.
The set $S$ is an \emph{independent cut} of $G$ if $S$ is both independent in $G$ and a cut of $G$.
We are interested in the following algorithmic problem.

\begin{problem}[H]
\noindent
\textsc{Independent Cut} (IC)\\
\textit{Instance:} A graph $G$.\\
\textit{Question:} Does $G$ have an independent cut?
\end{problem}

IC is \NP-complete even when the input is restricted to 5-regular line graphs of bipartite graphs~\cite{le2001onstablecutsets} or $K_4$-free planar graphs with maximum degree~5~\cite{le2008onstablecutsets}.
Since line graphs of bipartite graphs are perfect~\cite{brandstadt1999graphclasses}, 
IC remains \NP-complete on perfect graphs.
Numerous polynomial-time solvable cases have been identified~\cite{brandstadt2000onstablecutsets,le2001onstablecutsets,le2008onstablecutsets,chen2002note,pfender2013extremal,rauch2024revisitingextremalgraphshaving}.
Chen and Yu~\cite{chen2002note} showed that 
graphs of order $n$ and size at most $2n-4$
have an independent cut that can be found in polynomial time.
More polynomial-time solvable cases follow from the study of IC in the realm of parameterized complexity~\cite{rauch2025exactandparam,kratsch2025polynomialkernel}.
For example, IC is solvable in polynomial time on $P_5$-free graphs~\cite{rauch2025exactandparam}.

The \emph{order} of a graph $G$ is its number of vertices and the \emph{size} of $G$ is its number of edges.
The previously fastest known algorithm solving IC on a graph $G$ of order $n$ was due to Rauch et al.~\cite{rauch2025exactandparam} and had a runtime of $\calO^*\left(\sqrt[3]{3}^n\right)$.
Here, the $\calO^*$-notation suppresses polynomial factors in the $\calO$-notation.
It is based on the enumeration of all maximal independent sets, of which there are $\calO\left(\sqrt[3]{3}^n\right)$ many~\cite{moon1965oncliquesingraphs} and which can be generated with polynomial delay~\cite{johnson1988ongeneratingallmaximal}.

\paragraph{Results}
We present an algorithm that decides the IC on $G$ with an improved runtime of $\calO^*\left(\sqrt{2}^n\right)$.

Given an instance $G$ of IC, we may assume that $G$ is 2-connected; otherwise, $G$ has an independent cut.
Furthermore, we may assume that 
every vertex $u$ of $G$ lies in a triangle; otherwise, 
the neighborhood $N_G(u)$ of $u$ is an independent cut of $G$ 
since $G$ is 2-connected.
We call such an instance $G$ of IC \emph{non-trivial}.
Let $H$ be a graph.
We say that $H$ is \emph{solid} if it has order at least~3 and no independent cut.
For example, every triangle of $G$ is a solid subgraph of $G$.
For every independent set $S$ in $G$, 
each solid subgraph $H$ of $G$ is contained in 
$G[V(K) \cup S]$ for some component $K$ of $G-S$.
In fact, if $S$ is an independent set that is contained in the vertex set of the solid subgraph $H$, 
then $H-S$ is connected and has at least two vertices.
This implies that, if $H$ and $H'$ are solid subgraphs of $G$ and $H\cap H'$ contains an edge,
then $H\cup H'$ is a solid subgraph of $G$.
For a set $\calH$ of graphs, let $V(\calH) = \bigcup_{H \in\calH} V(H)$.
If the graphs in $\mathcal{H}$ are subgraphs of $G$, we say that $\mathcal{H}$ \emph{covers} $G$ if $V(G) = V(\calH)$, and $\calH$ \emph{quasi-covers} $G$ if $V(\mathcal{H})$ contains at least two vertices of every triangle of $G$.

The main idea of the algorithm is to guess a partition of a quasi-covering set of solid subgraphs of $G$ and to determine whether $G$ has a corresponding independent cut in polynomial time.
The precise statement is Theorem~\ref{thm:quasi-cover-algo}, 
which is our main result.

\begin{theorem}[restate=TheoremQuasicoverAlgorithm]\label{thm:quasi-cover-algo}
Given a non-trivial instance $G$ of order $n$ of \textup{IC}, and a quasi-covering set $\calH$ of solid subgraphs of $G$, we can decide in $2^{|\calH|}n^{\calO(1)}$ time whether $G$ has an independent cut or not.
\end{theorem}

\noindent
We postpone all proofs to Section~\ref{sec:proofs}.
Of course, Theorem~\ref{thm:quasi-cover-algo} alone is not enough, since we need a sufficiently small quasi-covering set of solid subgraphs.
This is the statement of Lemma~\ref{lem:quasi-cover-sparse}.

\begin{lemma}[restate=LemmaQuasicoverSparse]\label{lem:quasi-cover-sparse}
Given a graph $G$ of order $n$ and maximum degree at most $\Delta$, 
we can determine in polynomial time 
a quasi-covering set ${\cal H}$ of triangles of $G$ with 
$|{\cal H}|\leq \left( \frac{1}{2} - \alpha_\Delta \right)n$,
where $\alpha_\Delta = \frac{1}{2+4\lfloor \frac{\Delta}{2} \rfloor}$.
\end{lemma}

\noindent
Note that $\alpha_\Delta > 0$ and $\alpha_\Delta \rightarrow 0$ as $\Delta \rightarrow \infty$.
Theorem~\ref{thm:quasi-cover-algo} and Lemma~\ref{lem:quasi-cover-sparse} imply Corollary~\ref{cor:algo}, which contains the $\calO^*\left( \sqrt{2}^n \right)$ time algorithm as a special case.

\begin{corollary}[restate=CorollaryAlgorithm]\label{cor:algo}
Given a graph $G$ of order $n$ and maximum degree at most $\Delta$,
we can decide \textup{IC} on $G$ in $\calO^*\left(2^{\left( \frac{1}{2}-\alpha_\Delta \right)n}\right)$ time, 
where $\alpha_\Delta$ is as in Lemma~\ref{lem:quasi-cover-sparse}.
In particular, we can decide \textup{IC} on any graph of order $n$ in $\calO^*\left(\sqrt{2}^n\right)$ time.
\end{corollary}

\noindent
Since graphs of order $n$ and size at most $2n-4$ have an independent cut that can be found in polynomial time~\cite{chen2002note}, the statement of Corollary~\ref{cor:algo} becomes interesting only for $\Delta \geq 4$.
For $\Delta \in \{4,5\}$, we have $\frac{1}{2} - \alpha_\Delta = \frac{2}{5}$.
Hence, Corollary~\ref{cor:algo} shows that we can decide IC on graphs of $n$ and maximum degree at most $\Delta \in \{4,5\}$ in $\calO^*\left(2^{\frac{2}{5}n}\right)$ time.
Note that $2^\frac{2}{5} < 1.3196$.

As a complementary result, 
we observe that known constructions and arguments imply 
an asymptotic lower bound 
based on the Exponential Time Hypothesis (ETH) 
by Impagliazzo and Paturi~\cite{impagliazzo2001complexityofkSAT}. 

\begin{proposition}[restate=propLB]\label{prop:lb}
Unless \textup{ETH} fails, there is no algorithm that solves \textup{IC} on a graph of order $n$ and size $m$ in $2^{o(n+m)}$ time.
\end{proposition}

The following Lemma~\ref{lem:quasi-cover-dense}
shows that dense graphs 
have small quasi-covering sets 
of solid subgraphs.

\begin{lemma}[restate=LemmaQuasicoverDense]\label{lem:quasi-cover-dense}
Given a graph $G$ with $m$ edges 
such that every two adjacent vertices with a common neighbor 
have at least $\ell$ common neighbors, 
we can determine in polynomial time 
a quasi-covering set ${\cal H}$ 
of solid induced subgraphs of $G$ with 
$|{\cal H}|\leq\frac{2m}{(\ell+1)(\ell+2)}$.
\end{lemma}


\noindent
As a consequence of 
Lemma \ref{lem:quasi-cover-dense},
we obtain a fixed-parameter tractable algorithm for graphs of large minimum degree.

\begin{corollary}[restate=CorollaryAlgorithmDense]\label{cor:algo-dense}
Given a graph $G$ of order $n$ and minimum degree $\delta \geq \beta n$ 
for some $\beta > \frac{1}{2}$,
we can decide \textup{IC} on $G$ 
in $2^\frac{1}{(2\beta-1)^2}n^{\calO(1)}$ time.
In particular,
\textup{IC}
is fixed-parameter tractable under the parameter $\beta$.
\end{corollary}

\section{Proofs}\label{sec:proofs}
We start immediately with a proof of Theorem~\ref{thm:quasi-cover-algo}.

\TheoremQuasicoverAlgorithm*
\begin{proof}
We begin with some observations and a claim.

Assume that $G$ has an independent cut $S^*$.
By the properties of $G$ and $\cal H$,
there is a partition $(\calH_L,\calH_R)$ 
of $\cal H$ into two non-empty disjoint sets $\calH_L$ and $\calH_R$ that \emph{respects} $G-S^*$, 
meaning that there is no path in $G-S^*$
between a vertex in $V(\calH_L) \setminus S^*$ and a vertex in $V(\calH_R) \setminus S^*$.
Note that the sets 
$V(\calH_L) \setminus S^*$
and 
$V(\calH_R) \setminus S^*$
are both non-empty
and that $V(\calH_L) \cap V(\calH_R)$ 
is a subset of $S^*$,
which implies that $V(\calH_L) \cap V(\calH_R)$ is independent.

Every triangle of $G$ that
is not contained in $V(\calH)$
has exactly one vertex in $V(G) \setminus V(\calH)$ 
and is of one of two types:
Either it has two vertices in $V(\calH_L)$ or in $V(\calH_R)$, respectively, or it has one vertex in $V(\calH_L) \setminus V(\calH_R)$ and one vertex in $V(\calH_R) \setminus V(\calH_L)$; see Figure~\ref{fig:triangles}.
The following claim allows to extend $\calH$ and its partition $(\calH_L,\calH_R)$ with triangles of the first type.

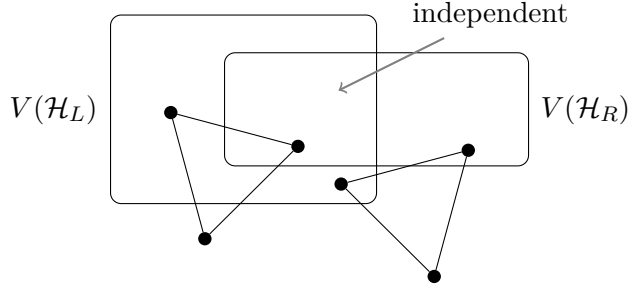
\begin{figure}
\centering
\begin{tikzpicture}
\draw[rounded corners] (-2.5,-1.5) rectangle (1,1);
\node at (-3.25,-0.25) {$V(\calH_L)$};
\draw[rounded corners] (-1,-1) rectangle (3,0.5);
\node at (3.75,-0.25) {$V(\calH_R)$};
\foreach \i in {1,2,3}{
	\node[dot,shift={(-1,-1)}] (u\i) at (120*\i+15:1) {};
}
\draw (u1) to (u2) to (u3) to (u1);
\foreach \i in {1,2,3}{
	\node[dot,shift={(1.5,-1.5)}] (v\i) at (120*\i+45:1) {};
}
\draw (v1) to (v2) to (v3) to (v1);
\node (ind) at (2.5,1) {independent};
\draw[->,thick,gray] (ind) to (0.5,0);
\end{tikzpicture}
\caption{A triangle of $G$ having exactly one vertex in $V(G) \setminus V(\calH)$ of each type.}\label{fig:triangles}
\end{figure}

\begin{claim}\label{claim:extend}
Let $C$ be a triangle of $G$ having exactly one vertex in $V(G) \setminus V(\calH)$.
\begin{enumerate}[(i)]
\item If $(\calH_L,\calH_R)$ respects $G-S^*$ and $C$ has two vertices in $V(\calH_L)$, then $(\calH_L \cup \{C\},\calH_R)$ respects $G-S^*$.
\item If $(\calH_L,\calH_R)$ respects $G-S^*$ and $C$ has two vertices in $V(\calH_R)$, then $(\calH_L,\calH_R \cup \{C\})$ respects $G-S^*$.
\end{enumerate}
\end{claim}
\begin{proof}
By symmetry, it suffices to prove (i).
If the triangle $C$ has two vertices in $V(\calH_L)$, one of these vertices cannot belong to the independent cut $S^*$.
Consequently, $V(C) \setminus S^*$ belongs to a component of $G-S^*$ that intersects 
$V(\calH_L) \setminus S^*$,
which completes the proof.
\end{proof}

We continue with a description of the algorithm.
It guesses a partition of $\calH$ into two non-empty disjoint sets $\calH_L$ and $\calH_R$.
If $I = V(\calH_L) \cap V(\calH_R)$ is not independent in $G$, we reject the guessed partition, 
because there is no independent cut $S^*$ of $G$ such that $(\calH_L,\calH_R)$ respects $G-S^*$.
Hence, suppose now that $I$ is independent.
For every vertex $v \in V(G) \setminus V(\calH)$, 
we fix a triangle of $G$ containing $v$, which is possible since $G$ is a non-trivial instance of IC.
Let $\calH'$ be the set of these triangles, and let $\calH'_L$ and $\calH'_R$ denote the triangles 
in $\calH'$ satisfying the requirements of Claim~\ref{claim:extend}~(i) and Claim~\ref{claim:extend}~(ii), respectively.
Let $\calL = \calH_L \cup \calH'_L$, $\calR = \calH_R \cup \calH'_R$, and $\calM = \calH' \setminus (\calH'_L \cup \calH'_R)$.
Repeated application of Claim~\ref{claim:extend} shows that, if $S^*$ is an independent cut of $G$, then $(\calH_L,\calH_R)$ respects $G-S^*$ if and only if $(\calL,\calR)$ respects $G-S^*$.
Note that 
\begin{itemize}
\item $\calL \cup \calR \cup \calM$ covers $G$, 
\item $\calL \cup \calR$ quasi-covers $G$, and \item $I = V(\calL) \cap V(\calR)$.
\end{itemize}
Let $L = V(\calL) \setminus V(\calR)$, $R = V(\calR) \setminus V(\calL)$, and $M = V(\calM) \setminus (V(\calL) \cup V(\calR))$.
Also note that every triangle in $\calM$ has exactly one vertex $u$ in $L$, 
exactly one vertex $v$ in $R$, and exactly one vertex $w$ in $M$.
In particular, this implies that, if $(\calL,\calR)$ respects $G-S^*$, where $S^*$ is an independent cut of $G$, then either $u$ or $v$ is in $S^*$, and, consequently, 
$w$ is not in $S^*$.

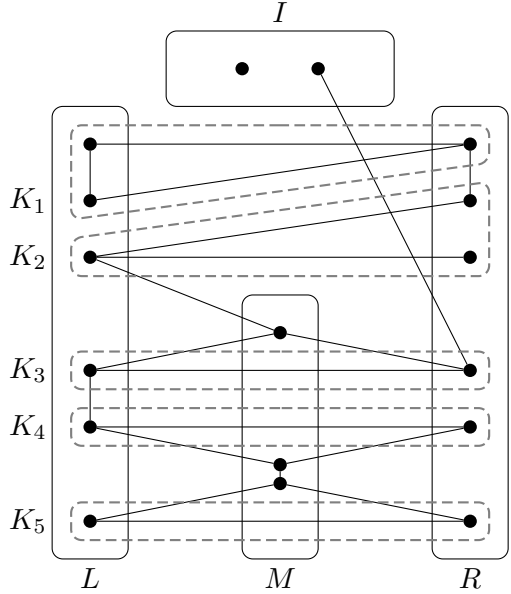
\begin{figure}
\centering
\begin{tikzpicture}
\draw[rounded corners] (-3,-3) rectangle (-2,3);
\node at (-2.5,-3.25) {$L$};
\draw[rounded corners] (3,-3) rectangle (2,3);
\node at (2.5,-3.25) {$R$};
\draw[rounded corners] (-1.5,3) rectangle (1.5,4);
\node at (0,4.25) {$I$};
\draw[rounded corners] (-0.5,0.5) rectangle (0.5,-3);
\node at (0,-3.25) {$M$};

\node[dot] (i1) at (-0.5,3.5) {};
\node[dot] (i2) at (0.5,3.5) {};
\node[dot] (m1) at (0,0) {};
\node[dot] (m2) at (0,-1.75) {};
\node[dot] (m3) at (0,-2) {};
\node[dot,label={[label distance=3.5mm]left:$K_3$}] (l1) at (-2.5,-0.5) {};
\node[dot,label={[label distance=3.5mm]left:$K_4$}] (l2) at (-2.5,-1.25) {};
\node[dot,label={[label distance=3.5mm]left:$K_5$}] (l3) at (-2.5,-2.5) {};
\node[dot] (r1) at (2.5,-0.5) {};
\node[dot] (r2) at (2.5,-1.25) {};
\node[dot] (r3) at (2.5,-2.5) {};

\draw (m1) to (l1) to (r1) to (m1);
\draw (m2) to (l2) to (r2) to (m2);
\draw (m3) to (l3) to (r3) to (m3);
\draw (m3) to (m2);

\node[dot] (L1) at (-2.5,2.5) {};
\node[dot,label={[label distance=3.5mm]left:$K_1$}] (L2) at (-2.5,1.75) {};
\node[dot,label={[label distance=3.5mm]left:$K_2$}] (L3) at (-2.5,1) {};
\node[dot] (R1) at (2.5,2.5) {};
\node[dot] (R2) at (2.5,1.75) {};
\node[dot] (R3) at (2.5,1) {};

\draw (L1) to (R1) to (L2) to (L1);
\draw (R3) to (L3) to (R2) to (R1);
\draw (L3) to (m1);
\draw (r1) to (i2);
\draw (l1) to (l2);

\draw[rounded corners,dash pattern=on 4pt off 2pt,gray,thick] (0,2.75) to (-2.75,2.75) to (-2.75,1.5) to (2.75,2.25) to (2.75,2.75) to (0,2.75);
\draw[rounded corners,dash pattern=on 4pt off 2pt,gray,thick] (0,0.75) to (-2.75,0.75) to (-2.75,1.25) to (2.75,2) to (2.75,0.75) to (0,0.75);
\draw[rounded corners,dash pattern=on 4pt off 2pt,gray,thick] (0,-0.25) to (-2.75,-0.25) to (-2.75,-0.75) to (2.75,-0.75) to (2.75,-0.25) to (0,-0.25);
\draw[rounded corners,dash pattern=on 4pt off 2pt,gray,thick] (0,-1) to (-2.75,-1) to (-2.75,-1.5) to (2.75,-1.5) to (2.75,-1) to (0,-1);
\draw[rounded corners,dash pattern=on 4pt off 2pt,gray,thick] (0,-2.25) to (-2.75,-2.25) to (-2.75,-2.75) to (2.75,-2.75) to (2.75,-2.25) to (0,-2.25);
\end{tikzpicture}
\caption{A schematic example of $G'$. The corresponding 2-SAT formula is\\$\calF = (x_1) \wedge (\bar{x}_1 \vee \bar{x}_2) \wedge (\bar{x}_2 \vee x_3) \wedge (x_3 \vee x_4) \wedge (\bar{x}_3) \wedge (\bar{x}_4 \vee x_5)$.}
\end{figure}

The following construction of the bipartite graph $G'$ and the 2-SAT formula $\calF$ is based on a construction by Rauch et al.~\cite{rauch2025exactandparam}.
Let $G'$ be the bipartite graph with
\begin{align*}
V(G') = L \cup R
\quad\text{and}\quad
E(G') = \left\{ uv \in E(G): u \in L\mbox{ and }v \in R \right\}.
\end{align*}
Furthermore, let $K_1, \dots, K_t$ be the vertex sets of the components of $G'$, and let $L_i = L \cap K_i$ and $R_i = R \cap K_i$.
For every $i \in [t]$, we introduce a Boolean variable $x_i$.
Let $\calF$ be the 2-SAT formula with variables $x_1,\dots,x_t$ containing the following clauses:
\begin{enumerate}[(i)]
\item For every $i \in [t]$ for which a vertex of $L_i$ has a neighbor in $I$, we add the clause $(x_i)$.\label{constr-i}
\item For every $i \in [t]$ for which a vertex of $R_i$ has a neighbor in $I$, we add the clause $(\bar{x}_i)$.\label{constr-ii}
\item For every (not necessarily distinct) $i,j \in [t]$ for which there is an edge $uv \in E(G)$ with $u \in L_i$ and $v \in L_j$, we add the clause $(x_i \vee x_j)$.\label{constr-iii}
\item For every (not necessarily distinct) $i,j \in [t]$ for which there is an edge $uv \in E(G)$ with $u \in R_i$ and $v \in R_j$, we add the clause $(\bar{x}_i \vee \bar{x}_j)$.\label{constr-iv}
\item For every component $K$ of $G[M]$ and for every $i,j \in [t]$ with $i \neq j$ for which $K$ has both a neighbor in $L_i$ and a neighbor in $R_j$, we add the clause $(\bar{x}_i \vee x_j)$.\label{constr-v}
\end{enumerate}
Note that we allow $i=j$ in \ref{constr-iii} and \ref{constr-iv}, 
in which case the 2-clause simplifies to a unit clause.
As the following claim shows,
the satisfiability of $\calF$ is linked to the existence 
of a specific independent cut of $G$.

\begin{claim}\label{claim:sat}
The \textup{2-SAT} formula $\calF$ is satisfiable if and only if there is an independent cut $S^*$ of $G$ such that $(\calL,\calR)$ respects $G-S^*$.
\end{claim}
\begin{proof}
Let $\alpha$ be a satisfying truth assignment of $\calF$.
Since $I$ is an independent set of $G$, $\calH$ is a set of solid subgraphs of $G$, and $(\calH_L,\calH_R)$ is a partition of $\calH$ where both $\calH_L$ and $\calH_R$ are non-empty, both $L$ and $R$ are not independent in $G$.
Therefore, there are $i,j \in [t]$ for which $(x_i \vee x_j)$ is a clause of $\calF$, and there are $i,j \in [t]$ for which $(\bar{x}_i \vee \bar{x}_j)$ is a clause of $\calF$, and, consequently, the assignment $\alpha$ is non-constant. 
We thus have $L \setminus S \neq \emptyset$ and $R \setminus S \neq \emptyset$ for
\[
S = \left(\bigcup_{i \in [t]:\:\alpha(x_i) = \texttt{false}} L_i\right) \cup \left(\bigcup_{i \in [t]:\:\alpha(x_i) = \texttt{true}} R_i\right).
\]
We claim that $S^*=S \cup I$ is an independent cut of $G$.
The clauses added to $\calF$ in \ref{constr-i}-\ref{constr-iv} ensure that $S^*$ is independent in $G$.
Assume, for a contradiction, that there is a path in $G-S^*$ between a vertex in $L_i$ and a vertex in $R_j$ for some $i,j \in [t]$, and that $P$ is a shortest such path.
Since by the construction of $S$, for every $k \in [t]$, 
either $L_k \subseteq S^*$ or $R_k \subseteq S^*$,
we must have $i \neq j$.
In particular, the path $P$ has length at least~2 and the internal vertices of $P$ are contained in a component $K$ of $G[M]$ for which $K$ has both a neighbor in $L_i$ and a neighbor in $R_j$.
Here, the clauses in \ref{constr-v} ensure that $L_i \subseteq S^*$ or $R_j \subseteq S^*$.
Thus, $P$ is not a path in $G-S^*$, a contradiction.
Altogether, it follows that $S^*$ is an independent cut of $G$
and that $(\calL,\calR)$ respects $G-S^*$.

Now, let $S^*$ be an independent cut of $G$ such that $(\calL,\calR)$ respects $G-S^*$.
We claim that, for every $i \in [t]$, either $L_i \subseteq S^*$ or $R_i \subseteq S^*$.
Indeed, otherwise we would have an edge with one end in $L_i$ and the other end in $R_i$ in $G-S^*$ by definition of $K_i$, a contradiction to the stated properties of $S^*$.
Therefore, the truth assignment
\[
\alpha(x_i) = \begin{cases}
\texttt{false} &\text{if $L_i \subseteq S^*$}\\
\texttt{true} &\text{if $R_i \subseteq S^*$}
\end{cases}
\quad \text{for $i \in [t]$},
\]
is well-defined.
We claim that $\alpha$ satisfies $\calF$.
Recall that $I \subseteq S^*$ as discussed before.
Consequently, $\alpha$ satisfies the clauses added to $\calF$ in \ref{constr-i}-\ref{constr-iv}, because $S^*$ is independent in $G$.
Let $i,j \in [t]$ with $i\neq j$ be such that a component $K$ of $G[M]$ has both a neighbor in $L_i$ and a neighbor in $R_j$.
In other words, $(\bar{x}_i \vee x_j)$ is a clause of $\calF$.
Then we must have $L_i \subseteq S^*$ or $R_j \subseteq S^*$.
It follows that $\alpha$ satisfies the clauses added to $\calF$ in \ref{constr-v}, too.
Altogether, $\alpha$ is a satisfying truth assignment of $\calF$.
\end{proof}

According to Claim~\ref{claim:sat}, 
we can decide if there is an independent cut $S^*$ of $G$ such that $(\calL,\calR)$ respects $G-S$ by deciding the satisfiability of the 2-SAT formula $\calF$, which is possible in polynomial time~\cite{aspvall1979lineartimealgorithm}.
Therefore, the overall runtime of the algorithm is $2^{|\calH|}n^{\calO(1)}$.
The correctness of the algorithm follows by the argumentation throughout the proof.
This completes the proof of Theorem~\ref{thm:quasi-cover-algo}.
\end{proof}

We now show how to construct the quasi-covering set of solid subgraphs needed in Theorem~\ref{thm:quasi-cover-algo}.


\LemmaQuasicoverSparse*
\begin{proof}
For a positive integer $p$, let a {\it $p$-windmill} 
be the graph that arises 
from $p$ disjoint copies of $K_2$ 
by adding a universal vertex. 
A graph is a {\it windmill} if it is a $p$-windmill for some $p$.
Note that a $1$-windmill is a triangle, and a $p$-windmill has $p$ triangles and $1+2p$ vertices.
A {\it ($p$-)windmill} in $G$ is a subgraph of $G$ that is a ($p$-)windmill.
A collection $\{ W_1,\ldots,W_k\}$ of vertex-disjoint windmills in $G$
is {\it maximal} if there is no collection $\{ W'_1,\ldots,W'_\ell\}$ 
of vertex-disjoint windmills in $G$ such that 
$W_1\cup\cdots\cup W_k$ is a proper subgraph of 
$W'_1\cup\cdots\cup W'_\ell$.
It is easy to see that every collection of vertex-disjoint windmills in $G$
can greedily be extended in polynomial time 
to some maximal collection of vertex-disjoint windmills in $G$.

Start with ${\cal W}$ 
initialized as any maximal collection of vertex-disjoint windmills in $G$,
and apply the following rule:
\begin{quote}
{\it If some triangle $T$ with exactly one vertex in $V({\cal W})$
intersects some $W$ from ${\cal W}$ in such a way that 
$W\setminus V(T)$ contains a triangle $T'$, 
then 
\begin{itemize}
\item replace $W$ within ${\cal W}$ 
by the two vertex-disjoint triangles $T$ and $T'$,
\item extend ${\cal W}$ 
to some maximal collection of vertex-disjoint windmills in $G$, and
\item denote the resulting collection again by ${\cal W}$.
\end{itemize}
}
\end{quote}
Repeated application of this rule results after polynomially many steps
in a maximal collection of vertex-disjoint windmills ${\cal W}$ in $G$
to which the rule no longer applies.
Let ${\cal W}=\{ W_1,\ldots,W_k\}$
and let $W_i$ be a $p_i$-windmill for $i\in [k]$.
Since $G$ has maximum degree at most $\Delta$,
we have $p_i\leq \left\lfloor\frac{\Delta}{2}\right\rfloor$ for $i\in [k]$.
Let ${\cal H}$ be the collection of 
the $p_1+\cdots+p_k$ triangles in $W_1\cup\cdots\cup W_k$.
It is easy to see that ${\cal H}$ quasi-covers $G$.
Furthermore,
\[
n\geq \sum\limits_{i=1}^k(1+2p_i)
\geq \frac{1+2\left\lfloor\frac{\Delta}{2}\right\rfloor}{\left\lfloor\frac{\Delta}{2}\right\rfloor}\sum\limits_{i=1}^kp_i
=\frac{1+2\left\lfloor\frac{\Delta}{2}\right\rfloor}{\left\lfloor\frac{\Delta}{2}\right\rfloor}
|{\cal H}|,
\]
which implies the desired bound on $|{\cal H}|$.
\end{proof}

We combine the statements of Theorem~\ref{thm:quasi-cover-algo} and Lemma~\ref{lem:quasi-cover-sparse} for a proof of Corollary~\ref{cor:algo}.

\CorollaryAlgorithm*
\begin{proof}
Since instances of IC
that are not non-trivial 
can be recognized in polynomial time
and all have independent cuts,
we may assume that 
$G$ is a non-trivial instance of IC.
Now, Theorem~\ref{thm:quasi-cover-algo} and Lemma~\ref{lem:quasi-cover-sparse} complete the proof.
\end{proof}

For the proof of the asymptotic lower bound of Proposition~\ref{prop:lb} we use Theorem~\ref{thm:eth-3sat}, which is an immediate consequence of the Exponential Time Hypothesis (ETH)~\cite{impagliazzo2001complexityofkSAT} and the Sparsification Lemma~\cite{impagliazzo2001whichproblems}.

\begin{theorem}[\hspace{1sp}\cite{impagliazzo2001complexityofkSAT,impagliazzo2001whichproblems}, see also~\cite{cygan2015parameterizedalgorithms}]\label{thm:eth-3sat}
Unless \textup{ETH} fails, \textup{3-SAT} cannot be solved in $2^{o(n+m)}$ time, where $n$ is the number of variables and $m$ is the number of clauses.
\end{theorem}

We are in a position to prove Proposition~\ref{prop:lb}.

\propLB*
\begin{proof}
For the proof, we describe a linear reduction
from 3-SAT to IC 
as the concatenation of
three known linear reductions
\begin{itemize}
\item from 3-SAT to \textsc{1-in-3-SAT},
\item from \textsc{1-in-3-SAT} to \textsc{1-in-3-SAT}, where every occurrence of a literal is positive, and
finally 
\item from \textsc{1-in-3-SAT}, where every occurrence of a literal is positive,
to IC.
\end{itemize}
The first two reductions are due to 
Schaefer~\cite{schaefer1978complexity}
as well as folklore constructions.
The last reduction is due to 
Brandstädt et al.~\cite{brandstadt2000onstablecutsets}.
We repeat these reductions to sufficient detail
in order to verify that they are indeed linear.

Let $R(a,b,c)$ denote the Boolean function that evaluates to \texttt{true} if exactly one of the literals $a$, $b$, or $c$ is \texttt{true} and to \texttt{false} otherwise.
Recall that, in an instance $\calF'$ of \textsc{1-in-3-SAT}, we are given a conjunction of the function $R$ on different literals, and the question is whether there exists a satisfying assignment of $\calF'$.
We refer to the instances of $R$ in an instance $\calF'$ of \textsc{1-in-3-SAT} as clauses of $\calF'$, too.

Let $\mathcal{F}$ be an instance of 3-SAT with $n$ variables and $m$ clauses.
For every clause $(x \vee y \vee z)$ of $\mathcal{F}$, we add the clauses
$R(a,b,\bar{x})$, $R(b,c,y)$, and $R(c,d,\bar{z})$
to a \textsc{1-in-3-SAT} formula $\calF'$, where $a$, $b$, $c$, and $d$ are new variables for every clause of $\calF$, respectively.
It follows that $\calF'$ has $n' = n + 4m$ variables and $m' = 3m$ clauses.
Clearly, we can construct $\calF'$ from $\calF$ in polynomial time.
It is easy to prove that $\calF$ is satisfiable if and only if $\calF'$ is satisfiable.

We further reduce $\calF'$ linearly to an instance of \textsc{1-in-3-SAT}, where every occurrence of a literal is positive.
We may assume that $\calF'$ contains an occurrence of a negative literal; otherwise, we are done.
We introduce three new variables $f_1$, $f_2$, and $t$ with the intention that $\alpha(f_1)=\alpha(f_2)=\texttt{false}$ and $\alpha(t)=\texttt{true}$ for every satisfying assignment $\alpha$.
For every variable $a$ of $\calF'$, we create a new variable $a'$, add the clauses $R(a,a',f_1)$ and $R(a,a',f_2)$, and, for every clause $R(\bar{a},b,c)$ of $\calF'$ in which $a$ occurs negatively, we successively replace $R(\bar{a},b,c)$ by $R(a',b,c)$.
Lastly, we add the clause $R(f_1,f_2,t)$.
Let $\calF''$ denote the thus obtained \textsc{1-in-3-SAT} instance.
Clearly, every literal in every clause of $\calF''$ is positive and we can construct $\calF''$ from $\calF'$ in polynomial time.
Furthermore, $\calF''$ has $n'' = 2n' + 3$ variables and $m'' = m' + 2n' + 1$ clauses, and $\calF'$ is satisfiable if and only if $\calF''$ is satisfiable.

\begin{figure}
\centering
\begin{tikzpicture}
\node[dot,label=left:$a_{i,1}$,shift={(-3,0)}] (a1) at (-120:1) {};
\node[dot,label=below:$a_{i,2}$,shift={(-3,0)}] (a2) at (0:1) {};
\node[dot,label=left:$a_{i,3}$,shift={(-3,0)}] (a3) at (120:1) {};
\node[dot,label=right:$b_{i,3}$,shift={(3,0)}] (b3) at (60:1) {};
\node[dot,label=below:$b_{i,2}$,shift={(3,0)}] (b2) at (120+60:1) {};
\node[dot,label=right:$b_{i,1}$,shift={(3,0)}] (b1) at (240+60:1) {};
\node[dot,label=above:$c_{i,1}$] (c1) at (0,0.866) {};
\node[dot,label=below:$c_{i,2}$] (c2) at (-0.5,0) {};
\node[dot,label=below:$c_{i,3}$] (c3) at (0.5,0) {};
\draw (a1) to (a2) to (a3) to (a1);
\draw (b1) to (b2) to (b3) to (b1);
\draw (c1) to (c2) to (c3) to (c1);
\draw (a1) to (b1);
\draw (a2) to (b2);
\draw (a3) to (b3);
\end{tikzpicture}    
\caption{The graph $G(C_i)$; cf. the proof of Proposition~\ref{prop:lb} and~\cite{brandstadt2000onstablecutsets}.}
\label{fig:GCi}
\end{figure}
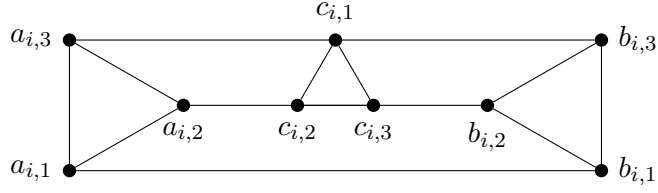

Now, we repeat a linear reduction from \textsc{1-in-3-SAT}, where every occurrence of a variable is positive, to IC by Brandstädt et al.~\cite{brandstadt2000onstablecutsets}.
For every variable $v$ of $\calF''$, we introduce a labeled vertex.
For every clause $C_i = R(c_{i,1},c_{i,2},c_{i,3})$ of $\calF''$, we define the labeled graph $G(C_i)$ as shown in Figure~\ref{fig:GCi}.
Moreover, we consider a triangle $R=r_1r_2r_3r_1$ and an edge $T=t_1t_2$.
We construct a graph $G$ from the vertices $v$ for every variable of $\calF''$, the graphs $G(C_i)$ for every clause $C_i$ of $\calF''$, the graphs $R$ and $T$, and the edges
\begin{itemize}
\item $vc_{i,j}$ if and only if $c_{i,j}$ is the variable $v$ for $i \in [m'']$ and $j \in [3]$,
\item $vr_1$, $vr_2$ for every variable $v$ of $\calF''$,
\item $r_1a_{i,1},r_2b_{i,1},r_3a_{i,1},r_3b_{i,1}$ for $i \in [m'']$, and
\item $t_1c_{i,1},t_1c_{i,2},t_2c_{i,1},t_2c_{i,3}$, for $i \in [m'']$.
\end{itemize}
Brandstädt et al.~\cite{brandstadt2000onstablecutsets} proved that $\calF''$ is satisfiable if and only if $G$ has an independent cut.
The graph $G$ has order $n'' + 9m'' + 5$ and size $14m'' + 4 + 3m'' + 2n'' + 4m'' + 4m''$ and can be constructed from $\calF''$ in polynomial time.

Altogether, the order and size of $G$ are linear in $n$ and $m$, respectively.
At this point, the statement of Proposition~\ref{prop:lb} follows from Theorem~\ref{thm:eth-3sat}.
\end{proof}

We proceed with a construction of a quasi-cover in the ``dense'' case, namely the proof of Lemma~\ref{lem:quasi-cover-dense}.

\LemmaQuasicoverDense*
\begin{proof}
For a vertex $u$ of $G$, 
let $t(u)$ be the number of non-trivial components of $G[N_G(u)]$.
If $v$ belongs to some non-trivial component $C$ of $G[N_G(u)]$,
then $u$ and $v$ have a common neighbor
and all at least $\ell$ common neighbors of $u$ and $v$ belong to $C$.
This implies that $C$ has order at least $\ell+1$, 
which implies 
\begin{eqnarray}\label{eorderh0}
t(u)&\leq &\frac{d_G(u)}{\ell+1}.
\end{eqnarray}
Start with ${\cal H}$ initialized as the set of all triangles in $G$,
and apply the following reduction rule to ${\cal H}$ as long as possible:
\begin{quote}
{\it If $H$ and $H'$ are distinct subgraphs in ${\cal H}$ 
and $H\cap H'$ contains an edge, then replace $H$ and $H'$
within ${\cal H}$ by their union, that is,
${\cal H}\longleftarrow \left({\cal H}\setminus \{ H,H'\}\right)\cup \{ H\cup H'\}$.}
\end{quote}
Now, let ${\cal H}$ denote the resulting set, 
once the reduction rule no longer applies.

The following properties are easy to see:
\begin{itemize}
\item Every graph in ${\cal H}$ is a solid induced subgraph of $G$.
\item Every triangle of $G$ belongs to some graph from ${\cal H}$, 
that is, ${\cal H}$ quasi-covers $G$.
\item ${\cal H}$ can be determined in polynomial time.
\item For every graph $H$ in ${\cal H}$ and every vertex $u$ of $H$,
some triangle in $H$ contains $u$.
\end{itemize}

Now, let $H$ be a graph in ${\cal H}$ and let $u$ be a vertex of $H$.

Let $uvw$ be a triangle of $H$ that contains $u$.
Let $C$ be the component of $G[N_G(u)]$ that contains the edge $vw$.
We claim that $V(C)\subseteq V(H)$.
In fact, suppose, for a contradiction, that this is not the case.
Then, since $C$ is a component of $G[N_G(u)]$, 
there are two adjacent vertices $x$ and $y$ in $C$ 
such that $x$ belongs to $H$ but $y$ does not.
Now, the triangle $uxy$ belongs to some graph $H'$ in ${\cal H}$
that is distinct from $H$, and the reduction rule could be applied
to $H$ and $H'$, which is a contradiction.
Hence, we obtain $V(C)\subseteq V(H)$ as claimed,
which implies 
\begin{eqnarray}\label{eorderh1}
n(H) & \geq & 1+|C|\geq \ell+2.
\end{eqnarray}
If $H$ and $H'$ are two distinct graphs in ${\cal H}$
with $\{ u\}\cup V(C)\subseteq V(H),V(H')$,
then the reduction rule could be applied
to $H$ and $H'$, which is a contradiction.
Hence, there is exactly one subgraph $H$ in ${\cal H}$
that contains $\{ u\}\cup V(C)$, which implies that 
\begin{eqnarray}\label{eorderh2}
\mbox{\it ${\cal H}$ contains at most $t(u)$ graphs that contain $u$.}
\end{eqnarray}
Now, double counting implies
\begin{eqnarray*}
|{\cal H}| 
&=&
\sum_{H\in {\cal H}}1\\
&\stackrel{(\ref{eorderh1})}{\leq}&
\frac{1}{\ell+2}\sum_{H\in {\cal H}}n(H)\\
&=&
\frac{1}{\ell+2}\sum_{u\in V(G)}|\{ H\in {\cal H}:u\in V(H)\}|\\
&\stackrel{(\ref{eorderh2})}{\leq}&
\frac{1}{\ell+2}\sum_{u\in V(G)}t(u)\\
&\stackrel{(\ref{eorderh0})}{\leq}&
\frac{1}{(\ell+1)(\ell+2)}\sum_{u\in V(G)}d_G(u)\\
&=& 
\frac{2m}{(\ell+1)(\ell+2)}.
\end{eqnarray*}
\end{proof}



Corollary~\ref{cor:algo-dense} 
is a consequence of 
Theorem~\ref{thm:quasi-cover-algo} and
Lemma~\ref{lem:quasi-cover-dense}. 

\CorollaryAlgorithmDense*
\begin{proof}
As in the proof of Corollary~\ref{cor:algo}, we may assume that $G$ 
is a non-trivial instance of IC.
If the vertices $u$ and $v$ are adjacent in $G$,
then $v$ has at least 
$\ell:=\delta-|V(G)\setminus N_G(u)|
\geq \delta-(n-\delta)
\geq (2\beta-1)n$
common neighbors with $u$.
Since $G$ has less than $n^2/2$ edges,
the quasi-cover ${\cal H}$ of $G$
from Lemma \ref{lem:quasi-cover-dense} 
satisfies $|{\cal H}|<\frac{n^2}{\ell^2}\leq \frac{1}{(2\beta-1)^2}$.
Now, the statement follows readily from Theorem~\ref{thm:quasi-cover-algo}.
\end{proof}

\end{document}